\documentclass[11pt,a4paper]{article}
\usepackage[qm]{qcircuit}

\usepackage{mathtools}
\usepackage{authblk} 
\usepackage{algpseudocode,algorithmicx,algorithm}

\usepackage{mathrsfs}
\usepackage{latexsym,bm}
\usepackage{amsmath,amsfonts,amssymb,amsthm}
\usepackage{extarrows}

\usepackage{graphicx,subfigure,epstopdf,float}
\usepackage{enumerate,cases,multirow}
\usepackage{makecell}
\usepackage{caption}

\usepackage{longtable,colortbl,arydshln,threeparttable}
\definecolor{mygray}{gray}{.9}

\usepackage{indentfirst}
\usepackage[top=25mm,bottom=20mm,left=25mm,right=20mm]{geometry}
\usepackage{cite}

\usepackage{listings}

\usepackage{makeidx}        
\usepackage{booktabs}
\usepackage[bookmarks,bookmarksnumbered,colorlinks,citecolor=red,linkcolor=red,hyperindex,linktocpage=true]{hyperref}

\newcommand{\ket}[1]{| #1 \rangle} 
\newcommand{\bra}[1]{\langle #1 |} 

\newcommand{\bb}{\boldsymbol}

\newcommand{\norm}[1]{\left\| #1 \right\|}
\def \d {\mathrm{d}}
\def \e {\mathrm{e}}
\def \i {\mathrm{i}}

\newcounter{parentalgorithm}

\makeatother

\newtheorem{theorem}{Theorem}[section]
\newtheorem{lemma}{Lemma}[section]

\newtheorem{definition}{Definition}[section]

\theoremstyle{remark}
\newtheorem{remark}{\bf Remark}[section]

\numberwithin{equation}{section}

\begin{document}

	\title{\bfseries Quantum algorithms for the fractional Poisson equation via rational approximation}

\author{Yin Yang\thanks{yangyinxtu@xtu.edu.cn}}
\author{Yue Yu\thanks{terenceyuyue@xtu.edu.cn}}
\author{Long Zhang}
\author{Ming Zhou}
\affil{School of Mathematics and Computational Science, Xiangtan University, Xiangtan, Hunan 411105, China}
\affil{Hunan Research Center of the Basic Discipline Fundamental Algorithmic Theory and Novel Computational Methods, Xiangtan, Hunan 411105, China}
\affil{National Center for Applied Mathematics in Hunan, Xiangtan, Hunan 411105, China}
	
\maketitle
	
\begin{abstract}
This paper presents a quantum algorithm for solving the fractional Poisson equation \((-\Delta)^s u = f\) with \(s \in (0,1)\) on bounded domains. The proposed approach combines rational approximation techniques with quantum linear system solvers to achieve exponential quantum advantage. The rational approximation represents the inverse fractional Laplacian as a weighted sum of standard resolvents, transforming the original nonlocal problem into a collection of shifted integer-order partial differential equations. These equations are consolidated into a single large linear system through a modified right-hand side construction that simplifies the quantum implementation. To enable practical implementation, we develop explicit quantum circuits via the Schr\"odingerization technique, which converts the non-unitary dynamics of the linear system into a higher-dimensional Schr\"odinger-type equation, allowing the use of standard Hamiltonian simulation. The circuit construction leverages the decomposition of shift operators to realize the discrete Laplacian and employs controlled operations to implement the select oracle. Under finite difference discretization, we provide detailed algorithmic procedures utilizing block-encoding techniques for the coefficient matrices. A comprehensive complexity analysis demonstrates that the quantum algorithm achieves a dependence on the inverse mesh size \(h^{-1}\) that is independent of the spatial dimension \(d\), in stark contrast to classical methods which suffer from exponential growth in high dimensions. This establishes an exponential quantum advantage for high-dimensional fractional problems, effectively overcoming the curse of dimensionality that limits classical approaches.
\end{abstract}



\section{Introduction}

Partial differential equations (PDEs) constitute a cornerstone of scientific computing, providing the mathematical framework for modeling diverse physical phenomena ranging from fluid dynamics and electromagnetic fields to quantum mechanical systems. As the demand for high-fidelity simulations intensifies across scientific and engineering disciplines, the computational burden associated with solving these equations becomes increasingly prohibitive—particularly in high-dimensional settings or when dealing with complex operators.

Within this landscape, fractional partial differential equations have emerged as a powerful modeling tool that extends beyond the capabilities of classical integer-order operators. By incorporating non-integer derivatives, these equations capture memory effects, anomalous transport, and long-range interactions that arise in complex physical systems. Among these, the fractional Laplacian $(-\Delta)^s$, where $s \in (0,1)$ denotes the fractional order, stands out as a fundamental nonlocal generalization of the standard Laplacian $-\Delta$. Over the past decades, this operator has attracted substantial research interest due to its ability to accurately describe phenomena characterized by nonlocality and long-range correlations. Its applications span a remarkable breadth of fields, including anomalous diffusion, L\'evy processes, stochastic dynamics, image processing, groundwater solute transport, financial mathematics, and turbulent flows \cite{Klafter2005Anomalous, Carreras2001Anomalous, Gatto2014Fractional, Vazquez2014Fractional, Du2019Nonlocal, Caffarelli2007Fractional, Nochetto2015Fractional}.

The inherent complexity of fractional operators, however, poses significant challenges for both analytical and numerical treatment. Exact solutions to fractional PDEs are rarely attainable, necessitating the development of robust numerical methodologies. Considerable effort has been devoted to discretizing problems involving fractional Laplacians, yielding a rich variety of approaches including finite element methods, spectral methods, Monte Carlo simulations, meshless techniques, finite difference schemes, and more recently, deep learning-based solvers \cite{Acosta2017Fractional, Acosta2017codeFractional, Yang2023fastNonlocal, Hao2020Spectral, Sheng2023MC, Hao2025Meshless, Hao2021fractionalFDM, YangChen2023fastNonlocalFDM, Gu2022DeepRitz}.

For the fractional Laplacian operator, there exist several distinct definitions of $(-\Delta)^s$ on bounded domains, including spectral definition and integral definition \cite{XuJC2023BPXfrational}, which are not mathematically equivalent. A fundamental difficulty arises from the nonlocal nature of fractional operators. Whereas standard PDE discretizations typically yield sparse matrices with favorable computational properties, fractional equations give rise to dense linear systems whose solution costs scale poorly with problem size. Specifically, regardless of which definition is adopted, and regardless of whether the finite element method or the finite difference method is used for discretization, the resulting stiffness matrices are dense.
 This ``curse of non-locality'' has driven intensive research into specialized numerical techniques capable of mitigating this complexity. Notable among these is the Caffarelli-Silvestre extension technique \cite{Caffarelli2007Fractional}, which transforms the integral fractional Poisson equation into a local elliptic PDE on a domain augmented by one additional dimension. Similarly, time-fractional gradient flows can be reformulated as integer-order problems on extended domains \cite{Eritz23timeFractional}, yielding systems of ordinary differential equations through discretization along both the extended and spatial coordinates. These extension-based approaches have recently inspired novel quantum algorithms built upon the Schr\"odingerization framework \cite{JLY22SchrShort, JLY22SchrLong}. We observe that a recent study \cite{An2026fractional} explores quantum algorithms for fractional reaction-diffusion equations. In that work, the spectral fractional Laplacian is discretized directly via Fourier spectral methods, and several time propagation approaches are examined, including second-order Trotter splitting, time-marching, truncated Dyson series, and a novel linear combination of Hamiltonian simulation in the interaction picture.

Quantum computing presents a fundamentally different paradigm that may circumvent some of these computational bottlenecks. A key advantage lies in its memory efficiency: representing a matrix of dimension $N$ and its associated solution vector requires only $\mathcal{O}(\log N)$ qubits~--~an exponential reduction compared to classical storage requirements. This logarithmic scaling has catalyzed growing interest in quantum algorithms for PDEs, particularly for large-scale simulations that overwhelm classical resources \cite{Cao2013Poisson, Berry2014Highorder, qFEM-2016, Costa2019Wave, Engel2019qVlasov, Childs2020qSpectral, Linden2020heat, Childs2021high, JinLiu2022nonlinear, GJL2022QuantumUQ, JLY2022multiscale, JLY22SchrShort, JLY22SchrLong, analogPDE}.

In this work, we develop a quantum algorithm for the spectral fractional Poisson equation based on rational approximation techniques. The rational approximation represents the inverse fractional Laplacian as a weighted sum of standard resolvents, thereby transforming the original nonlocal problem into a collection of shifted integer-order PDEs. A natural approach then solves each weighted problem independently and combines the results linearly. To fully exploit the advantages of quantum computing, we consolidate all equations into a single large linear system. By appropriately modifying the right-hand side vector, we significantly simplify the practical implementation of the algorithm. Under finite difference discretization, we provide detailed algorithmic procedures utilizing block-encoding techniques for the coefficient matrices. To enable practical implementation, we develop explicit quantum circuits for the resulting Hamiltonian evolution via the Schr\"odingerization technique. This approach converts the non-unitary dynamics of the linear system into a higher-dimensional Schr\"odinger-type equation, allowing the use of standard Hamiltonian simulation techniques. The circuit construction leverages the decomposition of shift operators to realize the discrete Laplacian and employs controlled operations to implement the select oracle. A comprehensive complexity analysis reveals that our quantum algorithm achieves an exponential advantage over classical methods, particularly in high-dimensional settings, as its complexity remains independent of the dimension \(d\) with respect to the inverse mesh size \(h^{-1}\).

The paper is structured as follows. In Section \ref{sec:fractionalPoisson}, we introduce the fractional Laplacian, provide a concise overview of the rational approximation method, establish relevant notation, and discuss its application to the fractional Poisson equation. Section \ref{sec:qalg} begins by presenting several block-encoding-based quantum subroutines, followed by a detailed description of the quantum algorithm for solving the fractional Poisson equation, along with a comprehensive analysis of its computational complexity. Section \ref{sec:circuit} develops explicit quantum circuits via the Schr\"odingerization technique, including detailed constructions of the Hamiltonian evolution and the select oracle. Conclusions are provided in the final section.

\section{Rational approximation of fractional Poisson equation} \label{sec:fractionalPoisson}

Let $\Omega$ be an open, bounded and connected subset of $\mathbb{R}^d$ ($d \geq 1$) with a Lipschitz boundary $\partial \Omega$.  Given $s \in (0, 1)$ and a sufficiently smooth function $f$, we consider quantum solvers for the fractional Poisson equation \cite{Caffarelli2007Fractional, Nochetto2015Fractional, LISCHKE2020109009}:
\begin{equation} \label{ProbFractional}
\begin{cases}
(-\Delta )^s u = f \qquad & \mbox{in}~~\Omega, \\
u = 0 \qquad & \mbox{on}~~\partial\Omega.
\end{cases}
\end{equation}

There exist several distinct definitions of the fractional Laplacian $(-\Delta)^{s}$ on bounded domains (see \cite{LISCHKE2020109009} for a comprehensive review).  The spectral fractional Laplacian is defined through the eigenvalues and eigenfunctions of the classical Laplacian $-\Delta$.  Let $\lambda_j$ denote the eigenvalues of $(-\Delta)$ and $e_j(x)$ the corresponding eigenfunctions.  Given a sufficiently smooth function $v(x)$, the spectral fractional Laplacian is defined as
\begin{equation}\label{sL}
(-\Delta)^{s}v(x)=\sum_j\lambda_j^{s}(v, e_j)_{L^2}e_j(x)
\end{equation}
where $(\cdot, \cdot)_{L^2}$ denotes the $L^2$ inner product on $[0, 1]^d$.  An alternative definition is the integral fractional Laplacian.
Given $s \in (0, 1)$, the integral fractional Laplacian of a function $ u:  \mathbb{R}^d \to \mathbb{R}$ is defined as
\[(-\Delta)^s u = C_{d, s} \text{p.v.} \int_{\mathbb{R}^d} \frac{u(x) - u(\xi)}{|x - \xi|^{d + 2s}} \d \xi, \]
where $ C_{d, s} = \frac{4^s s \Gamma(s + d/2)}{\pi^{d/2} \Gamma(1 - s)} $ is the normalization constant, and ``p.v.'' stands for the Cauchy principal value.  If $\Omega \subset \mathbb{R}^d$ ($d \geq 1$) is an open, bounded, connected set with a Lipschitz boundary, then the fractional Laplace operator on $\Omega$ is defined by zero-extending $u$ to $\mathbb{R}^d$.
For partial differential equations involving fractional Laplacians, direct discretizations~-~for instance, via finite difference or finite element methods~-~typically result in dense linear systems and high computational costs.  To overcome this challenge, efficient numerical methods have been developed that reformulate the spectral fractional diffusion problem using inverses of integer-order elliptic operators.  In this work, we focus on an approach based on rational approximations for the \textit{spectral} fractional Laplacian.

\subsection{Review of the rational approximation}

The literature features several rational approximation algorithms, including the Remez, BURA, and AAA methods.
A prominent example is the AAA algorithm, a seminal and highly efficient method that utilizes the barycentric representation
\[r(z)=\sum_{j=1}^m\frac{w_jf_j}{z-z_j}\Big/\sum_{j=1}^m\frac{w_j}{z-z_j}, \]
where $z_j$ are pre-selected support points, $f_j$ are function values at these points, and $w_j$ are weights.   Compared to the traditional Pad\'e approximation, the barycentric representation exhibits better numerical conditioning when appropriate support points are chosen, leading to reduced numerical errors.  Moreover, the selection of support points is flexible and independent of the poles of the rational function.  This flexibility facilitates the representation of a wide range of rational functions and serves as a foundation for constructing approximation algorithms.

For the inverse fractional Laplacian, it suffices to approximate the power function $x^{-s}$ with $0<s<1$ by a rational function of the form
\begin{equation}\label{ra}
 x^{-s} \approx r_{N_r}(x)=\sum_{\ell=1}^{N_r}\frac{c_{\ell}}{x+b_{\ell}} + c_{\infty}, \qquad x \in [\eta_*, \eta^*],
\end{equation}
where $\eta_*$ and $\eta^*$ are two positive constants and $b_{\ell}, c_{\ell}, c_{\infty} \ge 0$ (see \cite{Khristenko23fractional}). Without loss of generality, we may assume $b_1 \le b_2 \le \cdots \le b_{N_r}$, and $c_{\ell} \ne 0$; otherwise, terms with $c_{\ell} = 0$ can be omitted as they contribute nothing. This approach is particularly advantageous for solving fractional-order equations of the form $\mathcal{A}^su=f$. In this work, we adopt the rational approximation algorithm based on the AAA method introduced in \cite{Nakatsukasa2018AAA}, as it has been demonstrated in \cite{Khristenko23fractional} to exhibit notable efficiency and robustness across various test cases.

For the accuracy of \eqref{ra}, let us introduce some notations.  Let $I$ be a positive interval with left endpoint $\eta>0$ and let $X = L^{\infty}(I)$.
Let $\widetilde{D}$ be be a bounded set of $X$.  We define the symmetric convex hull of $\widetilde{D}$ as
\[B_1(\widetilde{D}) = \overline{\Big\{ \sum_{j=1}^m c_j g_j: g_j \in \widetilde{D}, \quad \sum_{i=1}^m |c_i| \le 1, \quad m \in \mathbb{N} \Big\}}. \]
Using this set, we introduce the variation norm on $X$ as
\[\|f\|_{\mathscr{L}_1(\widetilde{D})} = \mbox{inf} ~\Big\{ c>0:  f \in c B_1(\widetilde{D})\Big\}. \]
and the subspace of $X$ as
\[\mathscr{L}_1(\widetilde{D}) = \Big\{ f \in X:  \|f\|_{\mathscr{L}_1(\widetilde{D})} < \infty\Big\}. \]

\begin{lemma}[rational approximation, E.q~(6.1) in \cite{ Berrut2004BLI}] \label{lem:ra}
For any $ f \in \mathscr{L}_1(\widetilde{D}) $, there exists an algorithm that outputs $ b_{\ell} \ge 0, c_{\ell} > 0 $ for $ {\ell} = 1, \cdots, N_r $ and $c_{\infty} \ge 0$, such that
\[
\|f - \Pi_{N_r} f\|_{L^2(I)} \lesssim \|f\|_{\mathscr{L}_1(\widetilde{D})} K^{-{N_r}},
\]
where $K$ is an constant with $K > 1$, independent of $ f $ and $ N_r $,
\[
(\Pi_{N_r} f)(x) = \sum_{\ell=1}^{N_r} \frac{c_{\ell}}{f(x) + b_{\ell}} + c_{\infty}.
\]
\end{lemma}

\subsection{Solution given by rational approximation}\label{subsec:RA-FracLaplac}

For the fractional Laplacian $\mathcal{A}^s :=(-\Delta)^s$, where $\mathcal{A} = -\Delta$, according to Lemma \ref{lem:ra}, we can first seek a rational approximation of $x^{-s}$ in \eqref{ra} and then define the approximation
\begin{equation}\label{inversera}
u = \mathcal{A}^{-s}f = (-\Delta)^{-s} f \approx \Big(\sum_{{\ell}=1}^{N_r} c_{\ell} (\mathcal{A}+ b_{\ell} \mathcal{I})^{-1} + c_{\infty}\Big) f
=: \sum_{{\ell}=1}^{N_r} c_{\ell}u^{(\ell)} + c_{\infty} f,
\end{equation}
where $\mathcal{I}$ is the identity operator. For brevity, we set $N_r = 2^{n_r}$.
Suppose that $\mathcal{A}_h : \mathcal{V}_h \to \mathcal{V}_h$ is a discretization of $\mathcal{A}$ with maximum and minimum eigenvalues denoted by $\lambda_{\max} = \lambda_{\max}(\mathcal{A}_h)$ and $\lambda_{\min} = \lambda_{\min}(\mathcal{A}_h)$, respectively.  Then we can introduce the discretized version of \eqref{inversera}:
\begin{equation} \label{rah}
u_h:=\Big(\sum_{\ell=1}^{N_r} c_{\ell}(\mathcal{A}_h+ b_{\ell} \mathcal{I}_h)^{-1} + c_{\infty}\Big) f_h
=: \sum_{\ell=1}^{N_r} c_{\ell} u_h^{(\ell)} + c_{\infty} f_h,
\end{equation}
where $f_h$ is an approximation of $f$.

Direct discretizations, such as the finite difference method and the finite element method for \eqref{ProbFractional}, lead to dense linear systems. By utilizing the rational approximation \eqref{rah}, we can approximate the solution as a linear combination of $c_{\infty} f$ and numerical solutions to several shifted integer-order problems,
\begin{equation} \label{ProbSub}
\begin{cases}
-\Delta u^{(\ell)} + b_{\ell} u^{(\ell)} = f \qquad & \mbox{in}~~\Omega \\
u^{(\ell)} = 0 \qquad & \mbox{on}~~\partial\Omega
\end{cases}, \qquad \ell = 1, 2\cdots, N_r,
\end{equation}
thus overcoming the issue of dense matrix.

 In the following, we consider the finite difference discretizaition for \eqref{ProbSub} with $\ell$ fixed. For simplicity, let $v = u^{(\ell)}$.  We begin by considering the one-dimensional case with the domain taken as $\Omega = (a, b)$, which is uniformly divided into $M$ intervals of length $h = (b-a)/M$.  Let $v_i$ be the approximation of $v$ at $x_i = a + i h$ with $i=0, 1, \cdots, M$.  The unknown values are $v_1, \cdots, v_{M-1}$.  We apply the central difference discretization at $x = x_j$ for $j = 1, \cdots, M-1$:
\[
-\frac{v_{j-1} - 2v_j + v_{j+1}}{h^2} + b_{\ell} v_j = f_j := f(x_j), \quad j = 1, \cdots, M-1.
\]
Let $\bm{v}_h = [v_1, \cdots, v_{M-1}]^\top$.  Then one gets a linear system
\begin{equation}\label{linearsystem}
(A_h + b_{\ell} I_h) \bm{v}_h = \bm{f}_h,
\end{equation}
with
\begin{equation}\label{Ahfh}
A_h =  \frac{1}{h^2}
\begin{bmatrix}
2  &  -1       &           &      &    \\
 -1  & 2       & \ddots    &      &    \\
    &  \ddots  & \ddots    &  \ddots    &    \\
    &          & \ddots    & 2   & -1  \\
    &          &           &  -1   & 2 \\
\end{bmatrix}_{(M-1) \times (M-1)}
, \qquad
\bm{f}_h =
\begin{bmatrix}
f_1 \\
f_2  \\
\vdots\\
f_{M-2} \\
f_{M-1} \\
\end{bmatrix}.
\end{equation}
For brevity, we set $M-1 = 2^{n_x}$.

For $d$ dimensions, the domain is taken as $\Omega = (a, b)^d$.  To construct a spatial discretization, we  introduce $M+1$ spatial mesh points $a = x_{i, 0}<x_{i, 1}<\cdots<x_{i, M} = b$ by $x_{i, j} = a+j h$ in the $x_i$-direction, where $h = (b-a)/M$.  Let $\bm{j} = (j_1, j_2, \cdots, j_d)$, where $j_k=1, \cdots, M-1$ for $k=1, \cdots, d$.  With the central difference applied, we get
\[-\sum\limits_{k=1}^d \frac{ v_{\bm{j}-\bm{e}_k}-2 v_{\bm{j}} + v_{\bm{j}+\bm{e}_k}}{ h^2}
+ b_{\ell} v_{\bm{j}} = f_{\bm{j}} , \]
where $\bm{e}_k = (0, \cdots, 0, 1, 0, \cdots, 0)$ with the $k$-th entry being 1 and $j_k = 1, 2, \cdots, M-1$.   Denote by $\bm{v}_h$ the vector form of the $d$-order tensor $(v_{\bm{j}}) = (v_{j_1, j_2, \cdots, j_d})$:
\[\bm{v}_h = \sum\limits_{\bm{j} } v_{\bm{j}} \ket{\bm{j}}
= \sum\limits_{j_1, j_2, \cdots, j_d = 1}^{M-1} v_{j_1, j_2, \cdots, j_d} \ket{j_1, j_2, \cdots, j_d}, \]
where the index of $\ket{j_k}$ starts from $1$, which differs from the standard notation in quantum computation.
The associated linear system can be represented as
\[\sum\limits_{\bm{j}} \Big(\sum\limits_{k=1}^d \frac{ v_{\bm{j}-\bm{e}_k}-2 v_{\bm{j}} + v_{\bm{j}+\bm{e}_k}}{h^2}  + b_{\ell} v_{\bm{j}}\Big) \ket{\bm{j}} = \sum\limits_{\bm{j}} f_{\bm{j}} \ket{\bm{j}}. \]
As calculated in \cite{JLY22nonlinear, JLLY2024boundary}, one has
\begin{align*}
& \sum\limits_{\bm{j}} v_{\bm{j}-\bm{e}_k}\ket{\bm{j}} = \bm{v}_{0k} + ( I^{\otimes n_x} \otimes \cdots \otimes S^+  \otimes \cdots \otimes I^{\otimes n_x}) \bm{v}, \\
& \sum\limits_{\bm{j}} v_{\bm{j} + \bm{e}_k}\ket{\bm{j}}
 = \bm{v}_{Mk} + ( I^{\otimes n_x} \otimes \cdots \otimes S^-  \otimes \cdots \otimes I^{\otimes n_x}) \bm{v},
\end{align*}
where
\begin{align*}
& \bm{v}_{0k} = \sum\limits_{j_i = 1, i \ne k}^{M-1} u_{j_1, \cdots, 0, \cdots, j_d} \ket{j_1, \cdots, 1, \cdots, j_d} \quad (\mbox{1 is located at the $k$-th position}), \\
& \bm{v}_{Mk} = \sum\limits_{j_i = 1, i \ne k}^{M-1} u_{j_1, \cdots, M, \cdots, j_d} \ket{j_1, \cdots, M-1, \cdots, j_d} \quad (\mbox{$M-1$ is located at the $k$-th position})
\end{align*}
are the vectors generated by left and right boundary values in the $x_k$-direction, $S^+$ and $S^-$ are the shift operators defined by
\[S^+ = \begin{bmatrix}
0    &         &        &    &    \\
1    & 0       &        &    &   \\
     & \ddots  & \ddots &    &   \\
     &         &    1    &  0  &   \\
     &         &        & 1  & 0
\end{bmatrix}, \qquad S^- = (S^+)^\top. \]
For our problem, $\bm{v}_{0k} = \bm{v}_{Mk} = \bm{0}$.  The one-dimensional linear system in \eqref{linearsystem} is then replaced by
\begin{equation}\label{linearsystemd}
(A_{h, d} + b_{\ell} I_{h, d}) \bm{v}_{h, d} = \bm{f}_{h, d},
\end{equation}
where $\bm{v}_{h, d} = \bm{u}_{h, d}^{(\ell)}$ and
\begin{equation}\label{AhdFhd}
A_{h, d} = \sum_{\alpha = 1}^d (A_h)_\alpha, \qquad \bm{f}_{h, d} = \sum\limits_{\bm{j}} f_{\bm{j}} \ket{\bm{j}},
\end{equation}
with $ (\bullet)_\alpha := I^{\otimes (d-\alpha)} \otimes \bullet \otimes I^{\otimes (\alpha-1)}$.  The solution given by the rational approximation is
\begin{equation}\label{rationalsolution}
\bm{u}_h = \sum_{\ell=1}^{N_r} c_{\ell} \bm{u}_{h, d}^{(\ell)} + c_{\infty} \bm{f}_{h, d}.
\end{equation}

Throughout this article, we neglect the rational approximation error in Eq.~\eqref{inversera}.  Let $v(x_j)$ be the exact solutions of $v(x)$ at $x=x_j$ for \eqref{ProbSub} in one dimension.
We have
\[-\frac{v(x_{j-1}) - 2v(x_{j}) + v(x_{j+1})}{h^2}+ b_{\ell} v(x_j) = f(x_j) + r(x_j), \]
where
\[r(x_j) = - \frac{h^2}{12} v_{xxxx}(x_j) + \mathcal{O}(h^4). \]
Let $\bm{r}_h = [r(x_1), r(x_2), \cdots, r(x_{M-1})]^\top $.  Then the corresponding linear system for the exact solution is
\begin{equation}\label{exactlinearsystem}
	(A_h + b_{\ell} I_h) \bm{v} = \bm{f}_h + \bm{r}_h,
\end{equation}
where $\bm{v} = [v(x_1), \cdots, v(x_{M-1})]^\top$ is for $\bm{u}^{(\ell)}$ with $\ell = 1, \cdots, N_r$.

\begin{lemma}\label{lem:FDError}
Let $\Omega = (a, b)^d$ and let $v$ be the solution to \eqref{ProbSub}.  Assume that $v \in C^4(\overline{\Omega})$ and that the fourth-order derivative of $v$ is bounded in each direction.  Let $\bm{u}_h$ be the rational approximation solution defined in \eqref{rationalsolution}, and we denote $\bm{u}$ to be the exact counterpart.  Then there holds
\begin{equation}\label{FDerror}
\|\bm{u} - \bm{u}_h\|_h \lesssim \|\bm{c}\|_1 d h^2,
\end{equation}
where $\|\bm{c}\|_1 = \sum_{\ell = 1}^{N_r} |c_{\ell}| + |c_{\infty}|$, and
\[\|\bm{u} - \bm{u}_h\|_h = \frac{1}{\sqrt{M_d} } \|\bm{u} - \bm{u}_h\|\]
 is the discrete $L^2$ norm, with $M_d = (M-1)^d$ being the number of entries of $\bm{u}$.
\end{lemma}
\begin{proof}
We only consider the one-dimension case.
Let $\bm{e}_h = \bm{v} - \bm{v}_h = \bm{u}^{(\ell)} - \bm{u}_h^{(\ell)}$.  It is simple to find that $\bm{e}_h$ satisfies the following linear system:
\[(A_h + b_{\ell} I_h) \bm{e}_h = \bm{r}_h. \]
The matrix $A_h$ can be diagonalized as $A_h = PD P^\dag$, where $P$ is a unitary matrix and
\begin{equation}\label{DAh}
D= \text{diag}(d_1, \cdots, d_{M-1}), \qquad
d_j =  \frac{4}{h^2} \sin^2 \frac{\pi j}{2M} > 0.
\end{equation}
Then, one has
\begin{equation}\label{ErrorE}
\|\bm{u}^{(\ell)} - \bm{u}_h^{(\ell)}\| = \|\bm{e}_h\| =
\|P^\dag (D + b_{\ell} I)^{-1} P \bm{r}_h\|
\le \max_{i, {\ell}} \Big( \frac{1}{d_i + b_{\ell}} \Big) \|\bm{r}_h\| \lesssim \sqrt{M-1} h^2,
\end{equation}
where the last inequality holds since $8 \le d_i \le 4/h^2$ and $b_{\ell} \ge 0$, and we have used the assumption that $v \in C^4([a, b])$ and $|v_{xxxx}|$ is bounded.  By definition,
\[\|\bm{u} - \bm{u}_h\|
= \Big\|\sum_{\ell=1}^{N_r} c_{\ell} (\bm{u}^{(\ell)} - \bm{u}_h^{(\ell)})\Big\|
\le \sum_{\ell=1}^{N_r} |c_{\ell}| \|\bm{u}^{(\ell)} - \bm{u}_h^{(\ell)}\|
\le \|\bm{c}\|_1 \sqrt{M-1} h^2. \]
This completes the proof.
\end{proof}

The $N_r$ linear systems in \eqref{linearsystemd} can be combined into a large system
\begin{equation}\label{linearsystemmatrix}
H \bm{U}_h = \bm{F},
\end{equation}
where
\[\bm{U}_h = [\bm{u}_{h, d}^{(1)};~ \bm{u}_{h, d}^{(2)};~ \cdots;~ \bm{u}_{h, d}^{(N_r)}], \qquad
\bm{F} = [\bm{f}_{h, d};~ \bm{f}_{h, d};~ \cdots;~ \bm{f}_{h, d}], \]
with the semicolon ``;'' representing the column-wise vectorization, and
\begin{equation}\label{Hmatrix}
H = I^{\otimes n_r} \otimes A_{h, d} + B \otimes I^{\otimes {dn_x}}, \qquad B = \text{diag}(b_1, b_2, ~ \cdots, ~ b_{N_r}),
\end{equation}
where $I^{\otimes n_r}$ and $I^{\otimes {dn_x}} = I_{h,d}$ are $n_r$-qubit and $d n_x$-qubit identity matrices.

To facilitate the reconstruction of the solution to the original system, we can rewrite Eq.~\eqref{linearsystemmatrix} as
\begin{equation}\label{lsmVar}
	\tilde{H} \tilde{\bm{U}}_h = \tilde{\bm{F}},
\end{equation}
where $\tilde{H} = \text{diag}(H, I^{\otimes (n_r+{dn_x})})$, with $H$ defined in \eqref{Hmatrix}. The right-hand side is modified as
\begin{equation}\label{tildeF}
	\tilde{\bm{F}} := [c_1\bm{f}_{h, d};~ c_2\bm{f}_{h, d};~ \cdots;~ c_{N_r}\bm{f}_{h, d};~ c_{\infty}\bm{f}_{h, d};~ \underbrace{\bm{0};~ \cdots;~ \bm{0}}_{N_r-1}] =  \tilde{\bm{c}} \otimes \bm{f}_{h, d},
\end{equation}
where $\tilde{\bm{c}} = [c_1, c_2, \cdots, c_{N_r}, c_{\infty}, 0, \cdots, 0]^\top$ and $\bm{f}_{h, d}$ is given in \eqref{AhdFhd}. Here $N_r-1$ zeros vectors are appended to facilitate the block-encoding of the coefficient matrix.
It is evident that the modified linear system \eqref{lsmVar} gives
\begin{align}\label{tildesol}
	\tilde{\bm{U}}_h
	&:= [\tilde{\bm{u}}_{h, d}^{(1)};~ \tilde{\bm{u}}_{h, d}^{(2)};~ \cdots;~ \tilde{\bm{u}}_{h, d}^{(N_r)};~ \tilde{\bm{u}}_{h, d}^{(N_r+1)};~ \tilde{\bm{u}}_{h, d}^{(N_r+2)};~ \cdots;~ \tilde{\bm{u}}_{h, d}^{(2 N_r)}] \nonumber\\
	&= [c_1\bm{u}_{h, d}^{(1)};~ c_2\bm{u}_{h, d}^{(2)};~\cdots;~ c_{N_r}\bm{u}_{h, d}^{(N_r)};~ c_{\infty} \bm{f}_{h, d};~ \bm{0};~ \cdots;~ \bm{0}].
\end{align}
Therefore, the final approximate solution $\bm{u}_h$ can be represented as
\begin{equation}\label{specialLCU}
	\bm{u}_h = \tilde{\bm{u}}_{h} + c_{\infty} \bm{f}_{h, d}=\sum_{\ell=0}^{2N_r} \tilde{\bm{u}}_{h, d}^{(\ell)},
\end{equation}
where $\tilde{\bm{u}}_{h} = c_1\bm{u}_{h, d}^{(1)} + \cdots + c_{N_r}\bm{u}_{h, d}^{(N_r)}$.

\begin{lemma}\label{lem:H_inv}
Let $\tilde{H}$ be the coefficient matrix of the linear system \eqref{lsmVar}. Then the condition number of $H$ satisfies
\begin{equation}\label{KappaH}
\kappa_{\tilde{H}} = \| \tilde{H} \| \| \tilde{H}^{-1} \| = \mathcal{O} (\|\bm{b}\|_{\max} + d h^{-2} ),
\end{equation}
where $\|\bm{b}\|_{\max} = \max_{\ell}|b_{\ell}|$.
\end{lemma}
\begin{proof}
Let $d_i^*$ be an eigenvalue of $A_{h, d}$. One can check that
\[8d ~\le~ d_i^* ~\le~ \frac{4d}{h^2}, \qquad i=1,\cdots,M_d = (M-1)^d.\]
Since $B$ is a diagonal matrix and $0 \le b_1 \le \cdots \le b_{N_r}$, following the proof of Lemma \ref{lem:FDError}, one has
\begin{align*}
&\| \tilde{H}^{-1} \| \le \max\Big\{ \max_{i, {\ell}} \Big( \frac{1}{d^*_i + b_{\ell}} \Big), 1 \Big\}
\le \max\Big\{ \max_i \Big( \frac{1}{d^*_i + b_1} \Big), 1 \Big\}
\le \max\Big\{ \frac{1}{8d + b_1}, 1 \Big\} = 1 , \\
&\| \tilde{H} \| \le \max\Big\{\max_{i, {\ell}} ( d^*_i + b_{\ell} ), 1 \Big\}= \mathcal{O} (\|\bm{b}\|_{\max} + d h^{-2}),
\end{align*}
as required.
\end{proof}

\section{Quantum algorithm for the fractional Poisson equation} \label{sec:qalg}

\subsection{Block-encoding of the coefficient matrix}

The coefficient matrix of the linear system \eqref{linearsystem} or \eqref{linearsystemmatrix} is sparse, which allows us to construct it using the sparse access model \cite{Gilyen2019QSVD, Chakraborty2019blockEncode, Lin2022Notes}.  Let $A$ be a matrix with at most $s_r$ nonzero entries in any row and at most $s_c$ nonzero entries in any column, where we set $s_r = s_c = s$ for simplicity in what follows.  The sparse access model is described below.
\begin{definition}\label{def:sparsequery}
Let $A=(a_{ij})$ be an $n$-qubit matrix with at most $s$ non-zero elements in each row and column.  Assume that $A$ can be accessed through the following oracles:
\[O_r \ket{l} \ket{i} = \ket{r(i, l)} \ket{i}, \qquad O_c \ket{l} \ket{j} = \ket{c(j, l)} \ket{j},\]
\[O_A \ket{0}\ket{i, j} = \Big( a_{ij} \ket{0} + \sqrt{1 - |a_{ij}|^2} \ket{1} \Big) \ket{i, j}, \]
where $r(i, l)$  and $c(j, l)$ give the $l$-th non-zero entry in the $i$-th row and $j$-th column.
where $r(i, l)$ and $c(j, l)$ return the column index of the $l$-th non-zero entry in row $i$ and the row index of the $l$-th non-zero entry in column $j$, respectively.  The oracle $O_A$ can be equivalently replaced by
\[\tilde{O}_A \ket{0^a} \ket{i, j} = \ket{\tilde{a}_{ij}} \ket{i, j}, \]
where $\tilde{a}_{ij}$ represents an $a$-bit binary encoding of the matrix element $a_{ij}$.
\end{definition}

It is noteworthy that block-encoding provides a more generalized input model for matrix operations in quantum computing compared to sparse encoding \cite{Gilyen2019QSVD, Chakraborty2019blockEncode, Lin2022Notes, ACL2023LCH2, JLMY2025SchOptimal}.  This approach not only serves as an input model for quantum algorithms but also facilitates various matrix operations, thereby enabling the block-encoding of matrices with more complex structures.

\begin{definition}\label{def:blockencoding}
Given an $n$-qubit matrix $ A\in \mathbb C^{N\times N}$ with $N = 2^n$, and define the projection operator $ \Pi = \bra{0^m} \otimes I_n $, where $ I_n $ denotes the $ n $-qubit identity operator.  If there exist positive constants $ \alpha $ and $ \varepsilon $, and an $ (m+n) $-qubit unitary operator $ U_A $, such that the following inequality holds:
\[\| A - \alpha \Pi U_A \Pi^\dag \| = \| A - \alpha (\bra{0^m} \otimes I_n) U_A (\ket{0^m} \otimes I_n) \| \le \varepsilon,\]
then $ U_A $ is referred to as an $ (\alpha, m, \varepsilon) $-block-encoding of the matrix $ A $.
\end{definition}


The block-encoding of a sparse matrix can be constructed by leveraging its sparse access model \cite{Gilyen2019QSVD, Chakraborty2019blockEncode, Lin2022Notes, Lin2023timeMarching}.

\begin{lemma}[Block-encoding of sparse-access matrices] \label{lem:sparse2block}
Let $A=(a_{ij})$ be an $n$-qubit matrix with at most $s$ non-zero elements in each row and column.  Assume that  $\|A\|_{\max} = \max_{ij} |a_{ij}| \le 1$ and $A$ can be accessed through the sparse access model in Definition \ref{def:sparsequery}.  Then we have an implementation of $(s, n+3, \epsilon)$-block-encoding of $A$ with a single use of $O_r$ and $O_c$, two uses of $\tilde{O}_A$ and additionally using $\mathcal{O}(n+\log^{2. 5}(s/\epsilon))$ elementary gates and $\mathcal{O}(a+\log^{2. 5}(s/\epsilon))$ ancilla qubits, where $a$ represents $a$-bit binary representations of the entries of $A$.
\end{lemma}

If $O_r$, $O_c$ and $O_A$ are exact, then we can implement an $(s, n+1, 0)$-block-encoding of $A$ with a single use of $O_r$, $O_c$ and $O_A$ and additionally using $\mathcal{O}(n)$ elementary gates and $\mathcal{O}(1)$ ancilla qubits.

The following lemma describes how to perform matrix arithmetic on block-encoded matrices \cite{Gilyen2019QSVD, Chakraborty2019blockEncode}.

\begin{lemma}\label{lem:arithmeticBE}
Let $A_i$ be $n$-qubit matrix for $i=1, 2$.  If $U_i$ is an $(\alpha_i, m_i, \varepsilon_i)$-block-encoding of $A_i$ with gate complexity $T_i$, then
\begin{enumerate}
\item $A_1 + A_2$ has an $(\alpha_1+\alpha_2, m_1+m_2, \alpha_1\varepsilon_2 + \alpha_2\varepsilon_1)$-block-encoding that can be implemented with gate complexity $\mathcal{O}(T_1 + T_2)$.

\item $A_1 A_2$ has an $(\alpha_1\alpha_2, m_1+m_2, \alpha_1\varepsilon_2 + \alpha_2\varepsilon_1)$-block-encoding that can be implemented with gate complexity $\mathcal{O}(T_1 + T_2)$.

\item $A_1^\dag$ has an $(\alpha_1, m_1, \varepsilon_1)$-block-encoding that can be implemented with gate complexity $\mathcal{O}(T_1)$.

\item $A_1\otimes A_2$ has an $(\alpha_1\alpha_2, m_1+m_2, \alpha_1\epsilon_2 + \alpha_2\epsilon_1)$-block-encoding that can be implemented with gate complexity $\mathcal{O}(T_1+T_2)$.
\end{enumerate}
\end{lemma}

In the following, we proceed to construct the block-encoding of the coefficient matrixes $H$ and $\tilde{H}$ in \eqref{linearsystemmatrix} and \eqref{lsmVar}, respectively. For brevity, we assume that $M-1 = 2^{n_x}$ and $N_r = 2^{n_r}$, and that $\mathcal{U}_{B}$ and $\mathcal{U}_{A_h}$ are exact.

\begin{theorem}\label{thm:inputmodel}
Given an $(\alpha_{B}, m_B, \epsilon_B)$-block-encoding $\mathcal{U}_{B}$ of $B$ and an $(\alpha_{A_h}, m_{A_h}, \epsilon_{A_h})$-block-encoding $\mathcal{U}_{A_h}$ of $A_h$, where
\[\alpha_{B} \simeq \|\bm{b}\|_{\max}, \quad m_B = n_r+1, \quad \epsilon_B = 0, \]
\[\alpha_{A_h} \simeq h^{-2}, \quad m_{A_h} = n_x+1, \quad \epsilon_{A_h} = 0, \]
then we can construct an $(\alpha_{\tilde{H}}, m_{\tilde{H}}, \epsilon_{\tilde{H}})$-block-encoding, denoted by $\mathcal{U}_{\tilde{H}}$, of the coefficient matrix
\begin{equation}\label{tildeH}
	\tilde{H} = \operatorname{diag}(H, I^{\otimes (n_r+{dn_x})}) = |0\rangle\langle 0| \otimes H + |1\rangle\langle 1| \otimes I^{\otimes (n_r+{dn_x})},
\end{equation}
where $H = I^{\otimes n_r} \otimes A_{h,d} + B \otimes I^{\otimes {dn_x}}$ is the matrix appearing in the $d$-dimensional linear system \eqref{linearsystemmatrix}, with
\begin{align}
& \alpha_{\tilde{H}} = \mathcal{O}(\|\bm{b}\|_{\max} + dh^{-2}), \label{alphaH}\\
& m_{\tilde{H}} = d(d+1)n_x + 2n_r + d^2 + 2, \nonumber \\
& \epsilon_{\tilde{H}} = 0. \nonumber
\end{align}
This construction requires $\mathcal{O}(d)$ queries to $\mathcal{U}_{A_h}$ and $\mathcal{O}(1)$ queries to $\mathcal{U}_B$.
\end{theorem}
\begin{proof}
Lemmas \ref{lem:sparse2block} and \ref{lem:arithmeticBE} allow us to construct an $(\alpha_{B}, n_r+1, 0)$-block-encoding $\mathcal{U}_{B}$ of $B$ and an $(\alpha_{A_h}, n_x+1, 0)$-block-encoding $\mathcal{U}_{A_h}$ of $A_h$, where $\alpha_{A_h} \simeq h^{-2}$, with gate complexity $T_{A_h} =\mathcal{O}(n_x). $

We bound the complexity for $A_{h, d}$ in the following steps, where $A_{h, d}$ is defined in \eqref{linearsystemd}:
\begin{itemize}
\item It is clear that each term $(\bullet)_\alpha := I^{\otimes(d-\alpha)} \otimes \bullet \otimes I^{\otimes(\alpha-1)}$ in the summation
$A_{h, d} = \sum_{\alpha=1}^{d} (A_h)_\alpha$ admits an $(\alpha_{A_h}, d m_{A_h}, 0)$-block-encoding with gate complexity $T_{A_h}$.

\item By Lemma \ref{lem:arithmeticBE}, $A_{h, d}$ possesses an $(d \alpha_{A_h}, d^2 m_{A_h}, 0)$-block-encoding with gate complexity $dT_{A_h}$.
\end{itemize}

For the block-encoding of $H = I^{\otimes n_r} \otimes A_{h, d} + B \otimes I^{\otimes {dn_x}}$, we have the following steps:
\begin{enumerate}[(a)]
\item According to the previous discussion, the matrix $I^{\otimes n_r} \otimes A_{h, d}$ has a $(d \alpha_{A_h}, n_r+d^2 m_{A_h}, 0)$-block-encoding with gate complexity
\[T_{h, d} = dT_{A_h} = \mathcal{O}(d n_x). \]

\item The matrix $B \otimes I^{\otimes {dn_x}}$ has an $(\alpha_{B}, m_B + dn_x, 0)$-block-encoding. By Lemma \ref{lem:arithmeticBE}, $H = I^{\otimes n_r} \otimes A_{h, d} + B \otimes I^{\otimes {dn_x}}$ has an $(\alpha_H, n_H, \epsilon_H)$-block-encoding, where
\begin{align*}
	& \alpha_H = d \alpha_{A_h} + \alpha_{B}
	= \mathcal{O}(d h^{-2}+\|\bm{b}\|_{\max}), \\
	& m_H = n_r + d^2m_{A_h} + m_B + dn_x = d(d+1)n_x + 2n_r + d^2 + 1, \\
	& \epsilon_H = 0,
\end{align*}
with gate complexity $T_H = \mathcal{O}( d n_x+n_r )$.
\end{enumerate}

Since $\tilde{H}$ is defined as in \eqref{tildeH}, its block encoding can be realized via a controlled-$H$ operation. In other words, by Lemma~\ref{lem:arithmeticBE}, $\tilde{H}$ admits an $(\max\{\alpha_H, 1\}, m_H + 1, \epsilon_H)$-block-encoding, with gate complexity $T_{\tilde{H}} = T_H = \mathcal{O}(d n_x + n_r)$. Consequently, the gate complexity $T_{\tilde{H}}$ can be interpreted as $\mathcal{O}(d)$ queries to $\mathcal{U}_{A_h}$ and $\mathcal{O}(1)$ queries to $\mathcal{U}_B$, which completes the proof.
\end{proof}

\subsection{Recover the solution} \label{subsec:recover}

Before stating the main result, we first give the following optimal quantum algorithm for solving linear systems.
\begin{lemma}\cite{Costa2021QLSA}\label{lem:optQLSA}
Consider the linear system $A\bm{x} = \bm{b}$. Suppose we are given a block-encoding oracle for $A$ and a state-preparation oracle that generates the quantum state $\ket{\bm{b}}$.  Then there exists a quantum procedure that outputs the normalized solution state $\ket{A^{-1}\bm{b}}$ with error at most $\varepsilon$, by making
\begin{equation}
\mathcal{O}\bigl(\kappa_A \log(1/\varepsilon)\bigr)
\end{equation}
queries to the provided oracles, where $\kappa_A$ is the condition number of $A$.
\end{lemma}

Compared to the original system \eqref{linearsystemmatrix}, the formulation \eqref{specialLCU} corresponds to a special LCU structure, thereby simplifying the LCU procedure for the original system. In fact, let
\[\mathcal{D} = \begin{bmatrix}
	I^{\otimes dn_x} & I^{\otimes dn_x} & \cdots & I^{\otimes dn_x}\\
	* & * & \cdots & *\\
	\vdots & \vdots & \ddots & \vdots \\
	* & * & \cdots & *
\end{bmatrix}, \]
where $\mathcal{D}$ contains $2N_r = 2^{(n_r+1)}$ blocks per row and per column. It is evident that
\begin{equation}\label{DUh}
\mathcal{D} \tilde{\bm{U}}_h = [\bm{u}_h; ~*; \cdots; ~*],
\end{equation}
with the first block corresponding to \eqref{rationalsolution}. A simple choice of $\mathcal{D}$ is
\[\mathcal{D} = \sqrt{2N_r} \text{Had}^{\otimes (n_r+1)} \otimes I^{\otimes dn_x},\]
 where $\text{Had} = \dfrac{1}{\sqrt{2}} \begin{bmatrix} 1 & 1 \\ 1 & -1 \end{bmatrix}$ is the Hadamard gate.

For the approximate solution, we assume the following oracles:
\begin{itemize}
  \item the state-preparation oracles
  \[
O_{\tilde{c}}: \ket{0^{(n_r+1)}} \to \ket{\tilde{\bm{c}}}, \qquad
O_f: \ket{0^{dn_x}} \to \ket{\bm{f}_{h,d}};
\]
  \item  an $(\alpha_{A_h}, m_{A_h}, 0)$-block-encoding $\mathcal{U}_{A_h}$ of $A_h$ in \eqref{Ahfh}, where
\[\alpha_{A_h} \simeq h^{-2}, \quad m_{A_h} = n_x+1; \]

  \item an $(\alpha_{B}, m_B, 0)$-block-encoding $\mathcal{U}_{B}$ of $B$ in \eqref{Hmatrix}, where
\[\alpha_{B} \simeq \|\bm{b}\|_{\max}, \quad m_B = n_r+1. \]
\end{itemize}
From Eq.~\eqref{tildeF}, the state $\ket{\tilde{\bm{F}}} := \tilde{\bm{F}} / \|\tilde{\bm{F}}\|$ can be prepared as
\[
O_{\tilde{F}} := O_{\tilde{c}} \otimes O_f: \quad \ket{0^{(n_r+1)}}\ket{0^{dn_x}} \to \ket{\tilde{\bm{c}}} \otimes \ket{\bm{f}_{h,d}} = \ket{\tilde{\bm{F}}}.
\]
According to Theorem \ref{thm:inputmodel}, we can construct an $(\alpha_{\tilde{H}}, m_{\tilde{H}}, 0)$-block-encoding $\mathcal{U}_{\tilde{H}}$ for the coefficient matrix ${\tilde{H}}$ in \eqref{lsmVar}, where
\[\alpha_{\tilde{H}} = \mathcal{O}(\|\bm{b}\|_{\max} + dh^{-2}), \quad m_{\tilde{H}} = d(d+1)n_x + 2n_r + d^2 + 2,\]
utilizing $\mathcal{O}(d)$ queries to $\mathcal{U}_{A_h}$ and $\mathcal{O}(1)$ queries to $\mathcal{U}_B$.

For $\ket{{\tilde{H}}^{-1} \tilde{\bm{F}}} = \ket{\tilde{\bm{U}}_h}$, by Lemma \ref{lem:optQLSA}, there is a quantum algorithm, denoted by $\mathcal{U}_{\text{inv}}$, such that
\begin{equation}\label{Uunitary}
\ket{\tilde{\bm{F}}} \quad \xrightarrow{\mathcal{U}_{\text{inv}}} \quad  \ket{\tilde{\bm{U}}_h^q},
\end{equation}
where $\ket{\tilde{\bm{U}}_h^q}$ is an approximation of $\ket{\tilde{\bm{U}}_h}$ and satisfies
\begin{equation}\label{epsilon}
\|\ket{\tilde{\bm{U}}_h^q} - \ket{\tilde{\bm{U}}_h}\| \lesssim \varepsilon,
\end{equation}
with $\varepsilon$ to be determined. According to Lemmas \ref{lem:optQLSA} and \ref{lem:H_inv}, this makes
\begin{equation}\label{CH}
\mathcal{O}\Big(\kappa_{\tilde{H}} \log \frac{1}{\varepsilon}\Big) = \mathcal{O}\Big((\|\bm{b}\|_{\max} + d h^{-2}) \log \frac{1}{\varepsilon}\Big)=: \mathcal{O}(C_{\tilde{H}})
\end{equation}
queries to $O_{\tilde{F}}$ and $\mathcal{U}_{\tilde{H}}$. In terms of the provided oracles, this amounts to
\begin{itemize}
  \item $\mathcal{O}(C_{\tilde{H}})$ queries to $O_{\tilde{c}}$ and $O_f$,

  \item $\mathcal{O}(d C_{\tilde{H}})$ queries to $\mathcal{U}_{A_h}$,

  \item and $\mathcal{O}(C_{\tilde{H}})$ queries to $\mathcal{U}_B$.
\end{itemize}

For $\ket{\bm{u}_h}$, from \eqref{DUh} we get
\begin{align}\label{finalstate}
\ket{\tilde{\bm{U}}_h^q} \quad &\xrightarrow{\text{Had}^{\otimes (n_r+1)} \otimes I^{\otimes dn_x}} \quad  \frac{1}{\eta_1}\ket{0^{(n_r+1)}} \ket{\bm{u}_h^q} + \ket{\bot},
\end{align}
where $\ket{\bm{u}_h^q}$ is an approximation of $\ket{\bm{u}_h}$, $\ket{\bot}$ is orthogonal to $\ket{0^{(n_r+1)}} \ket{\bm{u}_h^q}$ and
\[\eta_1 = \frac{\sqrt{2N_r} \| \tilde{\bm{U}}_h^q\|}{\norm{\bm{u}_h^q}} \approx \frac{\sqrt{2N_r} \| \tilde{\bm{U}}_h\|}{\norm{\bm{u}_h}}.\]

For simplicity, we denote the above procedure as
\begin{equation}\label{Vfinal}
\ket{0^{(n_r+1)}}\ket{0^{dn_x}} \quad \xrightarrow{ V } \quad  \frac{1}{\eta_1}\ket{0^{(n_r+1)}} \ket{\bm{u}_h^q} + \ket{\bot},
\end{equation}
where $V = (\text{Had}^{\otimes (n_r+1)} \otimes I^{\otimes dn_x}) \mathcal{U}_{\text{inv}} O_{\tilde{F}}$.

The explicit time complexity necessitates an estimation of $\| \tilde{\bm{U}}_h\|$ and $\| \bm{u}_h\|$.
\begin{lemma}\label{lem:U_h_norm}
Let $\tilde{\bm{U}}_h$ be the solution to \eqref{lsmVar} and let $\bm{u}_h$ be defined in \eqref{specialLCU}. If $b_1 \le \cdots \le b_{N_r}$, where each $b_i\ge 0$, then there holds
\begin{align*}
\| \tilde{\bm{U}}_h\| \le \|\bm{c}\|_2 \|\bm{f}_{h,d}\|,   \qquad
\| \bm{u}_h\| \ge \frac{c_1}{\|\bm{b}\|_{\min} + 8d} \|\bm{f}_{h, d}\|,
\end{align*}
where $\|\bm{b}\|_{\min} = \min_{\ell}|b_{\ell}| = b_1$ and $\|\bm{c}\|_2 = (c_1^2 + \cdots + c_{N_r}^2 + c_{\infty}^2)^{1/2}$.
\end{lemma}
\begin{proof}
By Lemma \ref{lem:H_inv} and the definition of $\tilde{\bm{F}}$,
\begin{align*}
\| \tilde{\bm{U}}_h \| &= \| {\tilde{H}}^{-1} \tilde{\bm{F}} \| \le \| {\tilde{H}}^{-1} \| \| \tilde{\bm{F}} \|
\le \|\tilde{\bm{F}}\| = \|\tilde{\bm{c}} \otimes \bm{f}_{h,d}\| = \|\bm{c}\|_2 \|\bm{f}_{h,d}\|.
\end{align*}

Let $H_\ell := A_{h, d} + b_{\ell} I_{h, d}$. Then the solution of \eqref{linearsystemd} can be expressed as
$\bm{u}_{h, d}^{(\ell)} = H_\ell^{-1} \bm{f}_{h, d}$.
Since $c_{\infty} \ge 0, c_\ell > 0$, following the proof of Lemmas \ref{lem:FDError} and \ref{lem:H_inv}, one has
\begin{align*}
\|\bm{u}_h\|
= \|\tilde{\bm{u}}_h + c_{\infty} \bm{f}_{h,d}\|
\ge \min_i \left| \frac{c_1}{d_i^* + b_1} + \frac{c_2}{d_i^* + b_2} + \cdots + \frac{c_{N_r}}{d_i^* + b_{N_r}} \right| ~ \|\bm{f}_{h, d}\|.
\end{align*}
For a fixed $i$, set $a_j := \dfrac{1}{d_i^* + b_j}$ and define
\[S_k := \sum_{j=1}^k c_j, \qquad m_c := \min\{S_1,\cdots,S_{N_r}\}, \qquad M_c := \max\{S_1,\cdots,S_{N_r}\}.\]
Since $d_i^* > 0$ and $0 \le b_1 \le \cdots \le b_{N_r}$, we have $a_1 \ge a_2 \ge \cdots \ge a_{N_r} > 0$. By the Abel's inequality (see Theorem 1 in \cite{Mitrinovic1970}),
\[0 < m_c a_1 \le \sum_{k=1}^{N_r} c_k a_k \le M_c a_1,\]
that is,
\[\left| \frac{c_1}{d_i^* + b_1} + \frac{c_2}{d_i^* + b_2} + \cdots + \frac{c_{N_r}}{d_i^* + b_{N_r}} \right|
\ge  m_c a_1 = \frac{m_c}{8d + b_1}\]
for any fixed $i$. Since $c_i > 0$, we have $m_c = c_1$. This completes the proof.
\end{proof}

\begin{remark}
Using the empirical interpolation method (rEIM) from \cite{Aidi2025rEIM}, we can construct a rational approximation with $c_{\infty} = 0$. However, the resulting coefficients $c_i$ are not guaranteed to be positive. Our algorithm remains applicable, but this introduces difficulties in establishing a lower bound for $\| \bm{u}_h\|$ in Lemma~\ref{lem:U_h_norm}, which in turn complicates the estimation of the number of repetitions required in the subsequent procedure.
\end{remark}

\subsection{Time complexity analysis} \label{sec:complexity}

For the finite difference discretization, we analyze the error in the discrete $L^2$ norm, defined as
\[\norm{\bm{u}}_h := \Big( h^d \sum_{\bm{i}} u_{\bm{i}}^2 \Big)^{1/2} = h^{d/2} \norm{\bm{u}}, \]
where $\bm{i} = (i_1, \cdots, i_d)$ is a multi-index, and $h=(b-a)/M$.  In what follows we assume homogeneous boundary conditions.  For simplicity, we estimate the complexity primarily in terms of the number of queries to the oracle $\mathcal{U}_{A_h}$, which is the block-encoding of ${A_h}$.

Note we neglect the rational approximation error in \eqref{inversera}, only consider the error between the right-hand side of \eqref{inversera} and \eqref{rah} . The objective of this section is to quantify the computational cost of constructing a quantum solution $\bm{u}_h^q$ that satisfies
\[\norm{\bm{u}- \bm{u}_h^q}_h \leq \delta = \delta(h), \]
where
\begin{equation}\label{uhq}
\bm{u}_h^q = \sum_{\ell=1}^{2N_r} (\tilde{\bm{u}}_{h, d}^{(\ell)})^q,
\end{equation}
with $(\tilde{\bm{u}}_{h, d}^{(\ell)})^q$ being the quantum counterpart of \eqref{tildesol}.

The total error can be decomposed as
\[\norm{\bm{u} - \bm{u}_h^q}_h \leq \underbrace{\norm{\bm{u} - \bm{u}_h}_h}_{\text{discretization error}} \;+\; \underbrace{\norm{\bm{u}_h - \bm{u}_h^q}_h}_{\text{quantum approximation error}}. \]
We require each term on the right-hand side to be $\mathcal{O}(\delta)$, where $\bm{u}_h^q$ is derived from \eqref{linearsystemmatrix} by replacing the exact solution $\bm{U}_h$ with its quantum approximation $\bm{U}_h^q$.

The first error term arises from the finite-difference discretization.  According to the discussion in Lemma \ref{lem:FDError}, we have
\[ \norm{\bm{u} - \bm{u}_h}_h \lesssim \|\bm{c}\|_1 d h^2 =:\delta. \]
For the second term, we obtain from \eqref{linearsystemmatrix} that
\begin{align*}
\|\bm{u}_h - \bm{u}_h^q\|_h
& \simeq h^{d/2} \| \bm{u}_h- \bm{u}_h^q\|
  = h^{d/2}  \Big\| \sum_{\ell=1}^{2N_r} (\tilde{\bm{u}}_{h, d}^{(\ell)} - (\tilde{\bm{u}}_{h, d}^{(\ell)})^q) \Big\|  \\
& \le h^{d/2}  \sum_{\ell=1}^{2N_r} \Big\| \tilde{\bm{u}}_{h, d}^{(\ell)} - (\tilde{\bm{u}}_{h, d}^{(\ell)})^q\Big\|
  \lesssim h^{d/2} \|\tilde{\bm{U}}_h - \tilde{\bm{U}}_h^q\| .
\end{align*}
Requiring the right-hand side to be $\mathcal{O}(\delta)$, we can take
\begin{equation}\label{quantumerror}
\|\tilde{\bm{U}}_h - \tilde{\bm{U}}_h^q\| \lesssim  \frac{\delta}{h^{d/2}} = \| \bm{c} \|_1 d h^{2-d/2}.
\end{equation}

According to the discussion in Section \ref{subsec:recover}, we have the following main result.

\begin{theorem}
Let $\delta = \|\bm{c}\|_1 d h^2 $. Suppose that we are given the oracles in Section \ref{subsec:recover}. If $\|\bm{f}_{h, d}\|_h \approx \|f\|_{L^2} = \mathcal{O}(1)$, then there is a quantum algorithm that prepares an approximation of the normalized solution $\ket{\bm{u}_h}$, denoted as $\ket{\bm{u}^q_h}$, with $\Omega(1)$ probability and a flag indicating success, using
\begin{itemize}
  \item $\mathcal{O}(C_h)$ queries to $O_{\tilde{c}}$ and $O_f$,

  \item $\mathcal{O}(d C_h)$ queries to $\mathcal{U}_{A_h}$,

  \item and $\mathcal{O}(C_h)$ queries to $\mathcal{U}_B$,
\end{itemize}
where
\[ C_h = \mathcal{O}\Big( \sqrt{2N_r}(\|\bm{b}\|_{\max} + dh^{-2}) (\|\bm{b}\|_{\min} + 8d) c_1^{-1} \|\bm{c}\|_2 \log \frac{1}{h}\Big).  \]
The unnormalized vector $\bm{u}^q_h$ defined in \eqref{uhq} provides an $\mathcal{O}(\delta)$-approximation of $\bm{u}_h$ in the discrete $L^2$ norm.
\end{theorem}

\begin{proof}
Using Eq.~\eqref{quantumerror} and the inequality $\| \frac{\bm{x}}{\|\bm{x}\|} - \frac{\bm{y}}{\|\bm{y}\|} \| \le 2 \frac{\|\bm{x} - \bm{y}\|}{\|\bm{x}\|}$ for two vectors $ \bm{x}, \bm{y} $, we can bound the error defined in \eqref{epsilon} as
\[\| \ket{\tilde{\bm{U}}_h} - \ket{\tilde{\bm{U}}_h^q} \| \le \frac{2 \| \tilde{\bm{U}}_h - \tilde{\bm{U}}^q_h \|}{\| \tilde{\bm{U}}_h \|}
\lesssim \frac{\delta}{h^{d/2} \| \tilde{\bm{U}}_h \|}
= \frac{ d h^{2-d/2} \|\bm{c}\|_1}{\| \tilde{\bm{U}}_h \|}=:\varepsilon. \]
By Lemma \ref{lem:U_h_norm},
\begin{align*}
\frac{1}{\varepsilon}
 = \frac{\| \tilde{\bm{U}}_h \|}{ d h^{2-d/2}\|\bm{c}\|_1}
  \lesssim \frac{\|{\bm{c}}\|_2 \|\bm{f}_{h, d}\| }{ d h^{2-d/2}\|\bm{c}\|_1}
 = \frac{\|{\bm{c}}\|_2 \|\bm{f}_{h, d}\|_h }{ d h^{2} \|\bm{c}\|_1}\lesssim  h^{-2},
\end{align*}
where we have assumed $\|\bm{f}_{h, d}\|_h \approx \|f\|_{L^2} = \mathcal{O}(1)$. This yields
\[ C_{\tilde{H}} = \Big((\|\bm{b}\|_{\max} + d h^{-2}) \log \frac{1}{\varepsilon}\Big) = \mathcal{O}\Big( (\|\bm{b}\|_{\max} + d h^{-2}) \log \frac{1}{h}\Big),  \]
where $C_{\tilde{H}}$ is defined in \eqref{CH}.

Now we compute the repeated times for measurement. Notably, we apply $\text{Had}^{\otimes (n_r+1)} \otimes I^{\otimes dn_x}$ to the system register.  Upon measuring the first $(n_r+1)$-qubits of this register and obtaining the outcome $\ket{0^{\otimes (n_r+1)}}$, the remaining $dn_x$-qubits of the system register collapse to the desired state $\ket{\bm{u}_h}$. In view of \eqref{finalstate}, by applying the amplitude amplification \cite{BerryChilds2017ODE} and Lemma \ref{lem:U_h_norm}, we can get $\Omega(1)$ probability with $\mathcal{O}(\eta_1)$ repetitions of the combined oracle $V$ in \eqref{Vfinal}, where
\begin{align*}
\eta_1
 \approx \frac{\sqrt{2N_r} \| \tilde{\bm{U}}_h\|}{\norm{\bm{u}_h}}
\le \sqrt{2N_r} \|{\bm{c}}\|_2 \|\bm{f}_{h,d}\| \cdot \frac{\|\bm{b}\|_{\min} + 8d} {c_1 \|\bm{f}_{h, d}\|}
= \frac{\sqrt{2N_r} \|{\bm{c}}\|_2 (\|\bm{b}\|_{\min} + 8d)} {c_1}.
\end{align*}
This completes the proof.
\end{proof}

The quantum algorithm exhibits a dependence on $h^{-1}$ that is \textit{independent} of the spatial dimension $d$, in stark contrast to classical methods, which suffer from an exponential dependence as the dimension increases. This establishes an exponential quantum advantage for high-dimensional fractional problems, where classical approaches are fundamentally limited by the curse of dimensionality.

\section{Explicit quantum circuit via Schr\"odingerization} \label{sec:circuit}

\subsection{The Schr\"odingerization approach}

Since the coefficient matrix $\tilde{H}$ in \eqref{lsmVar} is positive definite, the solution to \eqref{lsmVar} can be viewed as the steady state of the linear ordinary differential equation system
\begin{equation}\label{reformulationODE}
	\frac{{\rm d} \bb{v}(t)}{{\rm d} t}  = 	- \tilde{H} \bb{v}(t) +  \tilde{\bm{F}},\quad
	\bb{v} (0) = \bb{0}.
\end{equation}
Denote $\bb{x}_T = \bb{v}(T)$ and $\bb{x} = \tilde{\bm{U}}_h$. Following the analysis in \cite{JLY2025fractionalPoisson}, one obtains
\begin{equation}\label{approxOperator}
 \|\bb{x}_T - \bb{x}\|  \le  \delta \|\bb{x}\|
\end{equation}
provided that $T \ge \frac{\kappa(\tilde{H})}{\|\tilde{H}\|} \log \frac{1}{\delta}$, where $\delta$ indicates the target accuracy. This reformulation makes it possible to apply the Schr\"odingerization technique to construct a Hamiltonian formulation suitable for quantum computation.

By introducing a time-independent auxiliary vector, \eqref{reformulationODE} can be rewritten as a homogeneous ODE system:
\begin{equation}\label{reformulationODEaug}
	\frac{{\rm d}}{{\rm d} t} \bb{v}_f = H_f \bb{v}_f,
	\quad
	H_f = \begin{bmatrix}
		-\tilde{H}  & \frac{I}{T}\\
		O & O
	\end{bmatrix},\quad
	\bb{v}_f(0) =  \begin{bmatrix}
		\bb{0} \\
		T \tilde{\bm{F}}
	\end{bmatrix},
\end{equation}
whose solution takes the form $\bb{v}_f = [\bb{v}^\top, T \tilde{\bm{F}}^\top]^\top$ (with $\top$ denoting transpose).

To carry out the Schr\"odingerization, we split $H_f$ into Hermitian and anti-Hermitian components: $H_f = H_1 + \i H_2$, where
\begin{equation}\label{eq:H12}
	H_1 = \frac{1}{2}(H_f + H_f^{\dagger}) = \begin{bmatrix}
		-\tilde{H}& \frac{1}{2T}I^{\otimes n_{\tilde{H}}} \\
		\frac{1}{2T}I^{\otimes n_{\tilde{H}}}  & O
	\end{bmatrix},\quad
	H_2 = \frac{1}{2\i}(H_f - H_f^{\dagger}) =  \frac{1}{2\i}
	\begin{bmatrix}
		O & \frac{1}{T} I^{\otimes n_{\tilde{H}}}\\
		-\frac{1}{T}I^{\otimes n_{\tilde{H}}} &O
	\end{bmatrix},
\end{equation}	
with $n_{\tilde{H}} = (n_r+1) + d n_x$. Both $H_1$ and $H_2$ are Hermitian, and $\dagger$ indicates conjugate transpose.

Applying the transformation $\bb{w}(t,p) = \e^{-p}\bb{v}_f(t)$ for $p\ge 0$ (extended naturally to negative $p$) yields
\begin{equation}\label{uf2v}
	\frac{\partial}{\partial t} \bb{w}(t) =  -H_1 \partial_p \bb{w} + \i H_2 \bb{w}.
\end{equation}
The initial condition is set as $\bb{w}(0,p) = \psi(p) \bb{v}_f(0)$ with $\psi(p) = \e^{-|p|}$.
When $H$ is positive semi-definite, the relation $\bb{v}_f(t)=\e^p \bb{w}(t,p)$ holds for all $p>p_* = \frac12$; thus $\bb{v}_f(t)$ can be recovered by projecting $\bb{w}$ onto the region $p > p_*$, and $\bb{v}(t)$ is then obtained by extracting the first block.

We restrict $p$ to the interval $[-\pi R, \pi R]$ with $R>0$ chosen so that $\e^{-\pi R} \approx 0$. A periodic boundary condition is applied in the $p$-direction, and a Fourier spectral discretization is employed. The auxiliary variable is discretized using a uniform mesh $\Delta p = 2 \pi R/N_p$ with an even number of points $N_p = 2^{n_p}$. The grid points are $-\pi R = p_0<p_1<\cdots<p_{N_p} = \pi R$. The one-dimensional basis functions for the Fourier spectral method are typically
\[\phi_l(p) = \e ^{\i \mu_l (p+\pi R)} , \quad \mu_l = \frac{l-N_p/2 }{R}, \quad l=0,1,\cdots,N_p-1.\]
We also define
\[\Phi = (\phi_{jl})_{N_p\times N_p} = (\phi_l(p_j))_{N_p\times N_p}, \qquad D_\mu = \text{diag} ( \mu_0, \cdots, \mu_{N_p-1} ).\]

Let $\bb{w}(t,p) = [w_1(t,p), \cdots, w_N(t,p)]^\top$ be the solution of \eqref{uf2v}. The spectral approximation for each component is
\[w_{i,h}(t,p) = \sum_{l=0}^{N_p-1} \tilde{w}_{i,l,h}(t) \phi_l(p), \qquad i = 1,\cdots,N,\]
where the subscript $h$ indicates the numerical solution. The corresponding vector is $\bb{w}_h(t,p) = [w_{1,h}(t,p),\cdots, w_{N,h}(t,p)]^\top$, which can be expressed as
\begin{equation}\label{interpcoeff}
\bb{w}_h(t,p) = \sum_{l=0}^{N_p-1} \tilde{\bb{w}}_{l,h}(t) \phi_l(p), \qquad
	\tilde{\bb{w}}_{l,h}(t) = \frac{1}{N_p} \sum\limits_{k=0}^{N_p-1} \bb{w}_h(t,p_k) \e ^{ - \i \mu_l (p_k+ \pi R)} .
\end{equation}
Define $\bb{W}_h$ as the collection of $\bb{w}_h$ evaluated at the grid points:
\[\bb{W}_h(t) = \sum_{k,i} w_{i,h}(t,p_k) \ket{k,i}
= [\bb{w}_h(t,p_0) ; \cdots; \bb{w}_h(t,p_{N_p-1})].\]
Similarly, let
\[\tilde{\bb{W}}_h(t) = \sum_{l,i} \tilde{w}_{i,l,h}(t) \ket{l,i}
= [\tilde{\bb{w}}_{0,h}(t) ; \cdots; \tilde{\bb{w}}_{N_p-1,h}(t)].\]
One can verify that $\bb{W}_h(t) = (\Phi \otimes I_{N\times N}) \tilde{\bb{W}}_h(t)$. The resulting ODE system becomes Hamiltonian:
\begin{equation} \label{eq:hamiltonian}
\frac{\rm d}{{\rm d} t} \bb{W}_h = -\i(P_\mu \otimes H_1) \bb{W}_h + \i (I^{\otimes n_p} \otimes H_2)\bb{W}_h, \qquad
\bb{W}_h(0) = \bb{\psi} \otimes \bb{v}_f(0),
\end{equation}
with $I^{\otimes n_p}$ representing the $n_p$-qubit identity and $\bb{\psi} = [\psi(p_0), \cdots, \psi(p_{N_p-1})]^\top$.
In terms of $\tilde{\bb{W}}_h = (\Phi^{-1} \otimes I_{N\times N})\bb{W}_h$, we obtain
\begin{equation}\label{discreteLCHSQLSP}
\begin{cases}
\dfrac{\d}{\d t} \tilde{\bb{W}}_h(t) = -\i (D_\mu \otimes H_1 - I^{\otimes n_p} \otimes H_2) \tilde{\bb{W}}_h(t), \\
\tilde{\bb{W}}_h(0) = \tilde{\bb{\psi}} \otimes \bb{v}_f(0), \quad \tilde{\bb{\psi}} = \Phi^{-1}\bb{\psi}.
\end{cases}
\end{equation}

The complete circuit for implementing the quantum simulation of $\ket{\bm{W}_h}$ is illustrated in Fig.~\ref{schr_circuit}, where
\begin{equation}\label{UT}
\mathcal{U}(T) = \e^{-\i (D_\mu \otimes H_1 - I^{\otimes n_p} \otimes H_2) T}.
\end{equation}
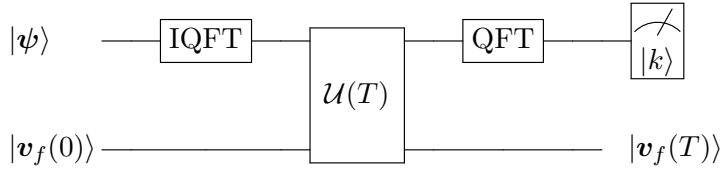
\begin{figure}[!htb]
	\centering
	\centerline{
		\Qcircuit @C=1em @R=2em {
			\lstick{\hbox to 2.7em{$\ket{\bm{\psi}}$\hss}}
			& \qw
			& \gate{\text{IQFT}}
			& 	\qw
			& \multigate{1}{\mathcal{U}(T)}	
			&  \qw
			& \gate{\text{QFT}}
			& \qw	
			& \qw  &\meterB{\ket{k}}\\
			\lstick{\hbox to 2.7em{$\ket{\bm{v}_f(0)}$\hss}}
			& \qw
			& \qw
			& \qw
			& \ghost{\mathcal{U}(T)}
			& \qw
			& \qw
			& \qw
			& \qw  & \hbox to 2em{$\ket{\bm{v}_f(T)}$\hss}
	}}
	\caption{Quantum circuit for the Schr\"odingerization method, where  $\bm{\psi} = \sum_{k=0}^{N_p-1} \psi(p_k)\ket{k} $.}
	\label{schr_circuit}
\end{figure}

\subsection{Quantum circuit for the Hamiltonian evolution}

The evolution operator for \eqref{UT} can be approximated via the first-order Lie-Trotter-Suzuki decomposition:
\[\mathcal{U}(T) = \prod_{j=1}^{N_t} \mathcal{U}(\Delta t), \qquad \Delta t = \frac{T}{N_t}, \]
where
\[\mathcal{U}(\Delta t) = \e^{-\i (D_\mu \otimes H_1 - I^{\otimes n_p} \otimes H_2) \Delta t}
\approx  \e^{\i (I^{\otimes n_p} \otimes H_2) \Delta t} \e^{-\i (D_\mu \otimes H_1 ) \Delta t}  =: V_2(\Delta t) V_1(\Delta t).\]

\textbf{The construction of $V_1(\Delta t)$}. Let $r = -\frac{\Delta t}{R}$. The unitary operator $V_1(\Delta t)$ can be written as
\begin{align*}
V_1(\Delta t)
 = \exp\Big(\sum_{k=0}^{N_p-1} \mu_k \ket{k}\bra{k} \otimes (-\i H_1 \Delta t) \Big)
 = \exp\Big(\sum_{k=0}^{N_p-1} \Big( k-\frac{N_p}{2} \Big) \ket{k}\bra{k} \otimes  \i r H_1  \Big) .
\end{align*}
Since the terms in the exponent commute, we have
\begin{align*}
V_1(\Delta t)
& = \sum_{k=0}^{N_p-1} \ket{k}\bra{k} \otimes ( \e^{\i r H_1 })^{k-N_p/2} \\
& = \Big(I^{\otimes n_p} \otimes ( \e^{\i r H_1 })^{-2^{n_p-1}} \Big)  \Big(\sum_{k=0}^{2^{n_p}-1} \ket{k}\bra{k} \otimes ( \e^{\i r H_1 })^k\Big) \\
& = \Big(I^{\otimes n_p} \otimes ( \e^{\i r H_1 })^{-2^{n_p-1}} \Big) \text{SEL}(\e^{\i r H_1}).
\end{align*}
Analogous to Fig.~2 of \cite{JLY2025fractionalPoisson}, the circuit implementation of $V_1(\Delta t)$, with $W =\e^{\i r H_1 }$, is presented in Fig.~\ref{fig:V1_detailed}.

\begin{figure}[htbp]
\centering
\centerline{
\Qcircuit @C=0.8cm @R=0.8cm {
    & \lstick{\ket{k_0}}       & \qw      & \ctrl{4}  & \qw     & \qw & \cdots  &   & \qw & \qw \\
    & \lstick{\ket{k_1}}       & \qw      & \qw       & \ctrl{3}& \qw & \cdots  &   & \qw & \qw \\
    \vdots \\
    & \lstick{\ket{k_{n_p-1}}} & \qw      & \qw       & \qw     & \qw & \cdots  &   & \ctrl{1} & \qw  \\
    &     & \gate{W^{-2^{n_p-1}}} & \gate{W^{2^{0}}} & \gate{W^{2^{1}}} & \qw &\cdots &   & \gate{W^{2^{n_p-1}}} & \qw
}}
\caption{Quantum circuit for $V_1(\Delta t)$. The first gate $W^{-2^{n_p-1}}$ corresponds to $(\e^{\i r H_1})^{-2^{n_p-1}}$, and the subsequent controlled-$W^{2^j}$ gates implement the select oracle $\text{SEL}(\e^{\i r H_1})$.}
\label{fig:V1_detailed}
\end{figure}
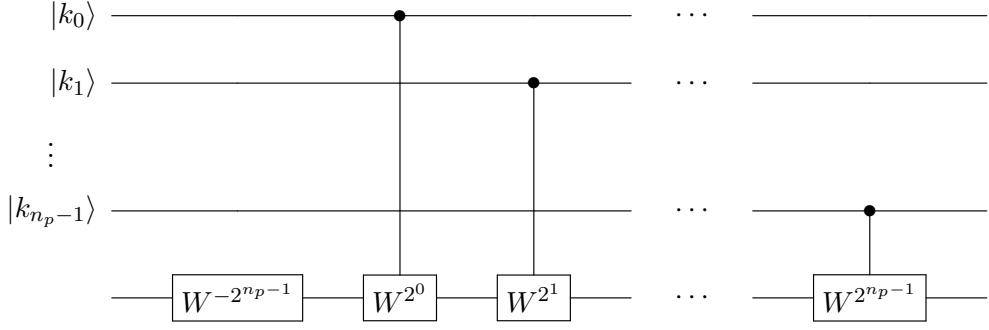

Now we need to approximate $\e^{\i r H_1}$. Using \eqref{eq:H12}, we decompose $H_1$ and apply the first-order Trotter decomposition to get
\begin{align*}
\e^{\i r H_1}
\approx \exp\Big( -\i r \begin{bmatrix}
		\tilde{H} &  O \\
		O  & O
	\end{bmatrix}  \Big)\exp\Big( \i r' \begin{bmatrix}
		 0& 1 \\
		1  & 0
	\end{bmatrix} \otimes I^{\otimes n_{\tilde{H}}}\Big)=: U_2 U_1,
\end{align*}
where $r' = \frac{r}{2T}$.  The matrix $U_1$ satisfies
\[U_1 = \exp( \i r' \sigma_x \otimes I^{\otimes n_{\tilde{H}}} ) = \e^{ \i r' \sigma_x }\otimes I^{\otimes n_{\tilde{H}}}
= R_x(-2r') \otimes I^{\otimes n_{\tilde{H}}}, \]
where $R_x(\theta) = \e^{-\i \theta \sigma_x/2}$ is the rotation gate about the $x$-axis. The matrix $U_2$ satisfies
\[U_2 = \begin{bmatrix}
		\e^{-\i r \tilde{H}} &  O \\
		 O  & I^{\otimes n_{\tilde{H}}}
	\end{bmatrix}
= \ket{0}\bra{0}\otimes \e^{-\i r \tilde{H}} + \ket{1}\bra{1} \otimes I^{\otimes n_{\tilde{H}}},\]
which is a controlled unitary operator.
The operator $\e^{-\i r \tilde{H}}$ can be written as
\begin{align*}
\e^{-\i r \tilde{H}}
& = \exp\Big(-\i r \begin{bmatrix}
H & \\
  & I^{\otimes n_H}
\end{bmatrix} \Big)
 = \ket{0}\bra{0} \otimes \e^{-\i r H}  + \ket{1}\bra{1} \otimes \e^{-\i r}I^{\otimes n_H} \\
& = \Big(\ket{0}\bra{0} \otimes \e^{-\i r H}  + \ket{1}\bra{1} \otimes I^{\otimes n_H}\Big)\Big( (\ket{0}\bra{0} + \e^{-\i r}\ket{1}\bra{1}) \otimes I^{\otimes n_H} \Big)  \\
& = \Big(\ket{0}\bra{0} \otimes \e^{-\i r H}  + \ket{1}\bra{1} \otimes I^{\otimes n_H}\Big)\Big(\text{Ph}(-r) \otimes I^{\otimes n_H} \Big) ,
\end{align*}
where $n_H = n_r+dn_x = n_{\tilde{H}}-1$ and $\text{Ph}(\theta) = \ket{0}\bra{0} + \e^{\i \theta} \ket{1}\bra{1}$ is the phase gate. This reduces to the realization of $\e^{-\i r H}$,which can be written as
\begin{align*}
\e^{-\i r H}
= \e^{-\i r (I^{\otimes n_r} \otimes  A_{h, d} +  B \otimes I^{\otimes {dn_x}})}
= (I^{\otimes n_r} \otimes \e^{-\i r A_{h,d}} ) ( \e^{-\i r B} \otimes I^{\otimes {dn_x}} ).
\end{align*}
The overall circuit for $\e^{-\i r \tilde{H}}$ is illustrated in Fig.~\ref{fig:exp_tilde_H}, where the explicit construction of $\e^{-\i r A_{h,d}}$ using the decomposition of shift operators is provided in \cite{HuJin24SchrCircuit,JLY24Circuits}, and $\e^{-\i r B}$ is a diagonal unitary operator.

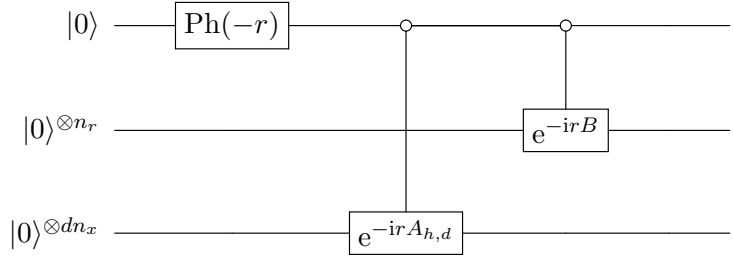
\begin{figure}[htbp]
\centering
\centerline{
\Qcircuit @C=0.8cm @R=0.8cm {
    & \lstick{\ket{0}}  & \gate{\text{Ph}(-r)} & \ctrlo{2} & \ctrlo{1} \qw \qw & \qw & \qw \\
    & \lstick{\ket{0}^{\otimes n_r}} & \qw & \qw & \gate{\e^{-\i r B}} & \qw & \qw \\
    & \lstick{\ket{0}^{\otimes {dn_x}}} & \qw& \gate{\e^{-\i r A_{h,d}}} & \qw & \qw & \qw 
}}
\caption{Quantum circuit for $\e^{-\i r \tilde{H}}  = \ket{0}\bra{0} \otimes \e^{-\i r H}  + \ket{1}\bra{1} \otimes \e^{-\i r}I^{\otimes n_H}$.}
\label{fig:exp_tilde_H}
\end{figure}

\textbf{The construction of $V_2(\Delta t)$}. Recall from \eqref{eq:H12} that
\[H_2 = \frac{1}{2T}
	\begin{bmatrix}
		0 & -\i \\
		\i & 0
	\end{bmatrix} \otimes I^{\otimes n_{\tilde{H}}} = \frac{1}{2T} \sigma_y \otimes I^{\otimes n_{\tilde{H}}},\]
where $\sigma_y$ is the Pauli-$Y$ matrix. Thus,
\begin{align*}
V_2(\Delta t)
= \e^{\i (I^{\otimes n_p} \otimes H_2) \Delta t} = I^{\otimes n_p} \otimes \e^{\i a \sigma_y} \otimes I^{\otimes n_{\tilde{H}}}, \qquad
a = \frac{\Delta t}{2T},
\end{align*}
Since $\e^{\i a \sigma_y} = R_y(-2a)$ with $R_y(\theta) = \e^{-\i \theta \sigma_y /2}$ being the standard rotation gate about the $y$-axis, $V_2(\Delta t)$ can be implemented by a single-qubit rotation.

\section{Conclusions}

We have developed a quantum algorithm for the fractional Poisson equation based on rational approximation. The method represents the inverse fractional Laplacian as a weighted sum of standard resolvents, transforming the nonlocal problem into shifted integer-order PDEs and yielding sparse linear systems amenable to quantum computation.
A key step is consolidating these multiple shifted problems into a single linear system via a modified right-hand side. This simplifies the quantum implementation and avoids explicit linear combination of independent solutions. Under finite difference discretization, we provide detailed procedures for block-encoding the coefficient matrices and recovering the final solution via Hadamard transformations.
Complexity analysis reveals a significant quantum advantage: the dependence on $h^{-1}$ is independent of the spatial dimension $d$, whereas classical methods suffer exponential growth. This dimension-independent scaling overcomes the curse of dimensionality, establishing an exponential quantum advantage that grows with dimension.

To enable practical implementation, we constructed explicit quantum circuits via the Schr\"odingerization technique, which converts the non-unitary dynamics of the linear system into a higher-dimensional Schr\"odinger-type equation, allowing the use of standard Hamiltonian simulation. The circuit construction leverages the decomposition of shift operators to realize the discrete Laplacian and employs controlled operations to implement the select oracle.

Future directions include extending the framework to more general fractional operators and boundary conditions, exploring alternative discretizations such as finite element or spectral methods, and investigating quantum preconditioning techniques to further improve complexity.

\section*{Acknowledgments}

Y Yang was supported by NSFC grant (No.\ 12571469), Scientific Research Innovation Capability Support Project for Young Faculty of China (No.\ SRICSPYF-BS2025132), the Project of Scientific Research Fund of the Hunan Provincial Science and Technology Department (No.\ 2024JJ1008).
Y Yu was supported by NSFC grant (No.\ 12301561), the Key Project of Scientific Research Project of Hunan Provincial Department of Education (No.\ 24A0100), the Science and Technology Innovation Program of Hunan Province (No.\ 2025RC3150) and the general program of Hunan Provincial Natural Science Foundation (No.\ 2026JJ50003).
L Zhang was supported by the Hunan Provincial Graduate Student Research and Innovation Project (No.\ CX20250933) and the Xiangtan University Graduate Student Research and Innovation Project (No.\ XDCX2025Y188).
This research was supported in part by the 111 Project (No.\ D23017), and Program for Science and Technology Innovative Research Team in Higher Educational Institutions of Hunan Province of China.

\bibliographystyle{plain} 
\bibliography{Refs}	

\end{document}